\documentclass[a4paper,10pt,twoside]{article}

\usepackage{amsmath}
\usepackage{amssymb}
\usepackage{amsthm}
\usepackage{pstricks}
\usepackage{pdfpages}

\makeatletter
\let\@afterindentfalse\@afterindenttrue
\@afterindenttrue
\makeatother

\newtheorem{theorem}{Theorem}
\newtheorem{lemma}{Lemma}

\newcommand{\di}{\displaystyle}
\renewcommand{\P}{\mathbb{P}}
\newcommand{\R}{\mathbb{R}}
\newcommand{\C}{\mathbb{C}}
\newcommand{\ci}{\mathop{\mathrm{i}}\nolimits}
\newcommand{\Rea}{\mathrm{Re}}
\newcommand{\Ima}{\mathrm{Im}}
\newcommand{\Q}{\mathcal{Q}}
\newcommand{\la}{\lambda}
\newcommand{\dint}{\,\mathrm{d}}
\newcommand{\M}[1]{\mathcal{M}_{#1}}
\newcommand{\Tr}{\mathop{\mathrm{Tr}}\nolimits}
\newcommand{\dg}{\textrm{diag}}
\newcommand{\cl}{\textrm{cl}}
\newcommand{\btop}[2]{\genfrac{}{}{0pt}{2}{#1}{#2}}

\newcommand*{\gz}[1]{\left( #1 \right)}
\newcommand*{\kz}[1]{\left\{ #1 \right\}}
\newcommand*{\sz}[1]{\left[ #1 \right]}
\newcommand*{\scprod}[1]{\left\langle #1 \right\rangle}
\newcommand*{\abs}[1]{\left\vert #1 \right\vert}
\newcommand*{\norm}[1]{\left\Vert #1 \right\Vert}

\title{
Volume of the space of qubit-qubit channels and state transformations under
  random quantum channels\thanks{
quantum channel, volume;
MSC2010: 81P16, 81P45, 94A17
}}

\author{Attila Lovas\thanks{lovas@math.bme.hu}, 
        Attila Andai\\
	Department of Analysis,\\
	Budapest University of Technology and Economics,\\
	1111 Budapest, Egry J\'ozsef street 1., Building H, Hungary
}

\date{\today}

\begin{document}

\maketitle 

\begin{abstract}
The simplest building blocks for quantum computations are 
  the qubit-qubit quantum channels. 
In this paper, we analyze the structure of these channels via their Choi representation. 
The restriction of a quantum channel to the space of classical states 
  (i.e. probability distributions) is called the underlying classical channel. 
The structure of quantum channels over a fixed classical channel is studied, 
  the volume of general and unital qubit channels with respect to the Lebesgue measure 
  is computed and explicit formulas are presented for the distribution of the 
  volume of quantum channels over given classical channels.
We study the state transformation under uniformly random quantum channels.
If one applies a uniformly random quantum channel (general or unital) to a given
  qubit state, the distribution of the resulted quantum states is presented.
\end{abstract}

\section{Introduction}

In quantum information theory, a qubit is the non-commutative analogue 
  of the classical bit.
A qubit can be represented by a $2\times 2$ complex self-adjoint positive semidefinite 
  matrix with trace one \cite{NielsenChuang,PetzQinf,RuskaiSzarekWerner}.
The space of qubits is denoted by $\M{2}$ and it can be identified with the unit ball 
  in $\R^3$ via the Stokes parameterization.
A linear map $Q:\M{2}\to \M{2}$ is called a qubit channel (or qubit quantum operation) 
  if it is completely positive and trace preserving (CPT) \cite{PetzQinf}.
A qubit channel is said to be \emph{unital} (or equivalently identity preserving)
  if it leaves the maximally mixed state invariant.

Choi and Jamio{\l k}owski has published a tractable representation for completely positive (CP) linear 
  maps \cite{Choi,Jamiolkowski}.
To a superoperator $Q:\C^{2\times 2}\to \C^{2\times 2}$ a block matrix
\begin{equation}
\label{eq:blk}
\begin{pmatrix} Q_{11} & Q_{12} \\  Q_{21} & Q_{22} \end{pmatrix}
  \quad Q_{11},Q_{12},Q_{21},Q_{22}\in\C^{2\times 2}
\end{equation}
  is associated, which is called Choi matrix, such that the action of $Q$ is given by
\begin{equation*}
\begin{pmatrix}  a & b\\  c & d \end{pmatrix}
  \mapsto aQ_{11} + bQ_{12} + cQ_{21} + dQ_{22}.
\end{equation*}
Due to Choi's theorem, the linear map $Q:\C^{2\times 2}\to\C^{2\times 2}$
  is CP if and only if its Choi matrix is positive definite \cite{Choi}. 
Hereafter, we will use the same symbol for the qubit channel and its Choi matrix. 
Clearly, a block matrix $Q$ of the form \eqref{eq:blk} corresponds to
  a qubit channel if and only if $Q_{11},Q_{22}\in\M{2}$, 
  $Q_{21} = Q_{12}^\ast$, $\Tr Q_{12}=0$ and $Q\ge 0$, thus
  the space of qubit channels can be identified with a
  convex subset of $\R^{12}$ which is denoted by $\Q$.
If we consider the set of unital qubit channels, 
  identity preserving property requires that 
  $Q_{11}+Q_{22}=I$ must hold in the Choi representation \eqref{eq:blk},
  hence the space of unital qubit channels ($\Q^1$)
  can be identified with a convex submanifold of $\R^9$.

Investigation of the set $\Q$ of all qubit channels play the 
  key role in the field of quantum information processing
  \cite{NielsenChuang}, since any physical transformation of a qubit
  carrying quantum information has to be described by an element of this set. 
Although the classical analogues of $\Q$ and $\Q^1$ are trivial objects, the geometric 
  properties of qubit channels are widely studied \cite{Omkar,AronPasieka}. 
However, the volume of the sets $\Q$ and $\Q^1$ is still unknown.
Random quantum operations and especially random qubit channels are subject of a 
  considerable scientific interest \cite{BoudaVarga}.
For example, an effect of external noise acting on qubits can be 
  modeled by random qubit channels. 
Authors in \cite{KZycz} have studied the spectral properties of quantum channels 
  and designed algorithms to generate random quantum maps.
We should mention that transformations of the maximally mixed state have important 
  applications in superdense coding \cite{HarrowA} which provide motivation 
  for research on the distance of the maximally mixed state and its image under
  the action of a random qubit channel.

Quantum channels are non-commutative analogues of classical stochastic maps, 
  therefore it is natural to consider their actions on classical quantum states 
  (i.e diagonal density matrices).
For a qubit channel $Q$, the \emph{underlying classical channel}
  is defined as the restriction of $Q$ to the space of classical bits.
For example, the following Markov chain transition matrix 
  represents the underlying classical channel of $Q\in\Q$ given by \eqref{eq:blk}
\begin{equation*}
P=\begin{pmatrix} \dg (Q_{11}) \\  \dg (Q_{22}) \end{pmatrix},
\end{equation*}
  where $\dg (Q_{ii})$ is a row vector that contains the diagonal of $Q_{ii}$. 

The main aim of this paper is to compute the volume of general and unital qubit channels 
  and investigate the distribution of the resulted quantum states if a general or unital
  uniformly random quantum channel was applied to a given state.
To compute the volume, we follow a similar strategy to those that was 
  introduced by Andai in \cite{AndaiVol} to calculate the volume of density matrices. 
This approach makes possible to gain information about the distribution
  of volume over classical states and to compute the effect of uniformly random
  quantum channels on the given state.

The paper is organized as follows.
In the second Section, we fix the notations for further computations and 
  we mention some elementary lemmas which will be used in the sequel.
In Section 3, the volume of general and unital qubit channels with respect 
  to the Lebesgue measure is computed and explicit formulas are presented 
  for the distribution of volume over classical channels.
Section 4 deals with state transformations under uniformly random quantum channels.

\section{Basic lemmas and notations}

The following lemmas will be our main tools, we will use them frequently.
We also introduce some notations which will be used in the sequel.

The first two lemmas are elementary propositions in linear algebra.
For an $n\times n$ matrix $A$ we set $A_{i}$ to be the left upper $i\times i$ 
  submatrix of $A$, where $i=1,\dots,n$.

\begin{lemma}
The $n\times n$ self-adjoint matrix $A$ is positive definite if and only if the inequality
  $\det(A_{i})>0$ holds for every $i=1,\dots,n$.
\end{lemma}

\begin{lemma}
Assume that $A$ is an $n\times n$ self-adjoint, positive definite matrix with entries 
  $(a_{ij})_{i,j=1,\dots, n}$ and the vector $x$ consists of the first $(n-1)$ elements 
  of the last column, that is $x=(a_{1,n},\dots,a_{n-1,n})$.
Then we have
\begin{equation*}
\det(A)=a_{nn}\det(A_{n-1})-\scprod{x,Tx},
\end{equation*}
  where $T=\det(A_{n-1}) (A_{n-1})^{-1}$.
\end{lemma}
\begin{proof}
Elementary matrix computation, one should expand $\det(A)$ by minors, 
  with respect to the last row.
\end{proof}

When we integrate on a subset of the Euclidean space we always integrate 
  with respect to the usual Lebesgue measure.
The Lebesgue measure on $\R^{n}$ will be denoted by $\la_{n}$.

\begin{lemma}
\label{le:integrateonellipsoid}
If $T$ is an $n\times n$ self-adjoint, positive definite matrix 
  and $k,\rho\in\R^{+}$, then
\begin{equation*}
\int\limits_{\kz{x\in\C^{n}\ \vert\ \scprod{x,Tx}<\rho}}
  (\rho-\scprod{x,Tx})^{k}\dint\la_{2n}(x)
  =\frac{\pi^{n} \rho^{n+k}k!}{(n+k)!\det T}.
\end{equation*}
\end{lemma}
\begin{proof}
The set $\kz{x\in\C^{n}\ \vert\ \scprod{x,Tx}<\rho}$ is an $n$ dimensional ellipsoid, 
  so to compute the integral first we transform our canonical basis to a new one, 
  which is parallel to the axes of the ellipsoid.
Since this is an orthogonal transformation, its Jacobian is $1$.
When we transform this ellipsoid to a unit sphere, the Jacobian of this transformation is
\begin{equation*}
\prod_{k=1}^{n}\frac{\rho}{\mu_{k}},
\end{equation*}
  where $(\mu_{k})_{k=1,\dots,n}$ are the eigenvalues of $T$.
Then we compute the integral in spherical coordinates.
The integral with respect to the angles gives the surface of the $2n$ dimensional sphere
  that is $\di \frac{2\pi^{n}}{(n-1)!}r^{2n-1}$.
The integral of the radial part is
\begin{equation*}
\int_{0}^{1}\frac{2\pi^{n}}{(n-1)!}r^{2n-1} \frac{\rho^{n}}{\det T} 
  (\rho-\rho r^{2})^{k} \dint r
  =\frac{2\pi^{n}\rho^{n+k}}{(n-1)!\det T}\int_{0}^{1}r^{2n-1}(1-r^{2})^{k}, 
  \dint r
\end{equation*}
  which gives back the stated formula. 
\end{proof}

\begin{lemma}
\label{lem:rv}
Assume that $X$ is a spherically symmetric and continuous random variable which 
  takes values in the unit ball $\{x\in\R^3 : \norm{x}\le 1\}$.
If $f$ denotes the probability density function of the $z$ component of $X$, then
  for the density of $\norm{X}$ we have
\begin{equation}
\label{eq:rv}
\rho(r) = -2r f'(r) \quad r\in \left]0,1\right[.
\end{equation}
\end{lemma}
\begin{proof}
Let us denote by $f_{X}$ the probability density function of $X$ in the unit ball.
The distribution of $X$ is rotation invariant thus
  there exists a function $g:\sz{0,1}\to\R^+$ such that, for every $x$ in the unit ball
  $f_{X}(x) = g(\norm{x})$.
For the radial density we have 
\begin{equation*}
\rho(r)=\frac{\dint}{\dint r}\P(\norm{X}<r) 
  =\frac{\dint}{\dint r} 4\pi\int\limits_0^r g(s)s^2\dint s
  =4\pi g(r)r^2 \quad r\in\left]0,1\right[.
\end{equation*}
Now we compute the density function of the $z$ component from the radial distribution.
For every $z_{0}\in\left]0,1\right[$
\begin{align*}
f(z_{0})&=\frac{\dint}{\dint z_{0}}\P (z<z_{0}) 
  =-\frac{\dint}{\dint z_{0}}\P (z\ge z_{0}) \\
&=-\frac{\dint}{\dint z_{0}}
  \int\limits_0^{2\pi}\int\limits_y^1\int\limits_0^{\arccos\gz{z_{0}/r}} g(r)r^2\sin\phi 
  \dint\phi\dint r\dint\theta
  =2\pi\int\limits_{z_{0}}^{1} g(r)r\dint r
\end{align*}
  holds and from this by derivation we get
\begin{equation*}
f'(r)=-2\pi g(r)r=-\frac{\rho(r)}{2 r},
\end{equation*} 
  which completes the proof.
\end{proof}

\section{The volume of qubit channels}

To determine the volume of different qubit quantum channels
  we use the same method which consists of three parts.
First, we use a unitary transformation to represent channels
  in a suitable form for further computations.
Then we split the parameter space into lower dimensional
  parts such that the adequate applications of the previously
  mentioned lemmas lead us to the result.

\subsection{General qubit channels}

The following parametrization of $\Q\subset\R^{12}$ is considered
\begin{equation} 
\label{eq:matQ}
Q=\begin{pmatrix}
a_{1} & b & c & d\\
\bar{b} & a_{2} & e & -c\\
\bar{c} & \bar{e} & f_{1} & g\\
\bar{d} & -\bar{c} & \bar{g} &f_{2} \end{pmatrix},
\end{equation}
  where $a_{1},f_{1}\in\sz{0,1}$, $a_{2}=1-a_{1}$, $f_{2}=1-f_{1}$ and $b,c,d,e,g\in\C$.
Let us define $a=a_{1}$ and $f=f_{1}$.

The underlying classical channel corresponding to these parameter values are given by 
$Q_{\cl}=\begin{pmatrix}
        a & 1-a \\
        f & 1-f
       \end{pmatrix}$.
Let us choose the unitary matrix
\begin{equation}
\label{eq:unit}
U=\begin{pmatrix}
1&0&0&0\\
0&0&1&0\\
0&1&0&0\\
0&0&0&1 \end{pmatrix}
\end{equation}
  and define the matrix $A$ as
\begin{equation}
\label{eq:matA}
A=U^{*}QU=\begin{pmatrix}
a&c&b&d\\
\bar{c}&f&\bar{e}&g\\
\bar{b}&e&a_{2}&-c\\
\bar{d}&\bar{g}&-\bar{c}&f_{2}\end{pmatrix},
\end{equation}
  which is positive definite if and only if $Q$ is positive definite.

\begin{theorem}
The volume of the space $\Q$ with respect to the Lebesgue measure is
\begin{equation*}
V(\Q)=\frac{2\pi{^5}}{4725}
\end{equation*}
  and the distribution of volume over classical channels can be written as
\begin{equation}
\label{eq:Vaf}
V(a_{1},f_{1})=\frac{2^{4}\pi^5}{45}
\left\{ \begin{array}{cll}
\di a_{1}^{3}f_{1}^{3}
  (a_{1}^2f_{1}^2-5a_{1}a_{2}f_{1}f_{2}+10a_{2}^2f_{2}^2) 
  & \mbox{if} & a_{1}+f_{1}\leq 1, \\[1em]
\di a_{2}^{3}f_{2}^{3}
  (a_{2}^2f_{2}^2-5a_{1}a_{2}f_{1}f_{2}+10a_{1}^{2}f_{1}^2)
  & \mbox{if} & a_{1}+f_{1}>1.
\end{array}
\right.
\end{equation}
\end{theorem}
\begin{proof}
Since there is an unitary transformation \eqref{eq:unit} between 
  the set of matrices of the form of \eqref{eq:matA} and the quantum channels
  their volumes are the same.
We compute the volume of the set of matrices given by parameterization \eqref{eq:matA}.

The volume element corresponding to the parametrization \eqref{eq:matQ} is 
  $2^{7}\dint\la_{12}$.
The matrix $A$ in Equation \eqref{eq:matA} is positive definite if and only if 
  $\det (A_i)>0$ for $i=1,2,3,4$. 

First we assume that the parameters $a,f$ and the submatrix $A_{3}$ are given and 
  consider the requirement $\det A_{4}\geq 0$.
Simple calculation shows that we have 
\begin{equation*}
\det A_{4}=R_{3}-
  \scprod{\begin{pmatrix} d'\\g'\end{pmatrix},
  T_{3} \begin{pmatrix} d'\\g'\end{pmatrix}},
\end{equation*}
  where $\di d'=d+\frac{bc}{a_{2}}$, $\di g'=g+\frac{c\bar{e}}{a_{2}}$, 
  $\di R_{3}=\gz{\det A_{3}}\gz{f_{2}-\frac{\abs{c}^{2}}{a_{2}}}$ and
\begin{equation*}
T_{3}=\begin{pmatrix} 
  a_{2}f_{1}-\abs{e}^{2} & be-a_{2}c\\ \bar{b}\bar{e}-a_{2}\bar{c} & a_{1}a_{2}-\abs{b}^{2}
 \end{pmatrix}.
\end{equation*}
In this case the inequality $\det A_{4}\geq 0$ can be written in the form of
\begin{equation}
\label{ineqR}
\scprod{\begin{pmatrix} d'\\g'\end{pmatrix},
  T_{3} \begin{pmatrix} d'\\g'\end{pmatrix}}\leq R_{3}.
\end{equation}
The matrix $T_{3}$ is positive, because $\det T_{3}=a_{2}\det A_{3}\geq 0$ 
  and $\gz{T_{3}}_{11}$ is the determinant of the middle $2\times 2$ submatix of $A$.
It means, that the Inequality \eqref{ineqR} has solution if, and only if
  $f_{2}a_{2}\geq \abs{c}^{2}$.

The transformation $(d,g)\mapsto(d',g')$ is a shift, therefore it does
  not change the volume element.
We have by Lemma \ref{le:integrateonellipsoid}
\begin{equation*}
V(a,f,b,c,e)=2^{7}\int\limits_{
  \scprod{\begin{pmatrix} d'\\g'\end{pmatrix},
  T_{3} \begin{pmatrix} d'\\g'\end{pmatrix}}\leq R_{3}} 1 \dint (d',g')
=\frac{2^{6}R_{3}^{2}\pi^{2}}{\det T_{3}},
\end{equation*}
  where $\det T_{3}=a_{2}\det A_{3}$, therefore
\begin{equation*}
V(a,f,b,c,e)=
\left\{\begin{array}{cll}
\di\frac{2^{6}\pi^{2}}{a_{2}^{3}}
  \gz{a_{2}f_{2}-\abs{c}^{2}}^{2} \det A_{3} 
  &\mbox{if} & f_{2}a_{2}\geq \abs{c}^{2},\\
\di 0  &\mbox{if} & f_{2}a_{2}<\abs{c}^{2}.
\end{array}
\right.
\end{equation*}

In the second step we assume that the parameters $a,f$ 
  and the submatrix $A_{2}$ are given and consider the requirement $\det A_{3}\geq 0$.
We have 
\begin{equation*}
\det A_{3}=R_{2}-
  \scprod{\begin{pmatrix} b\\ \bar{e}\end{pmatrix},
 T_{2} \begin{pmatrix} b\\ \bar{e}\end{pmatrix}},
\end{equation*}
  where $R_{2}=(1-a)\det A_{2}$ and
\begin{equation*}
T_{2}=\begin{pmatrix}  f & -c\\ -\bar{c} & a \end{pmatrix}.
\end{equation*}
The inequality $\det A_{3}\geq 0$ can be written in the form of
\begin{equation*}
\scprod{\begin{pmatrix} b\\ \bar{e}\end{pmatrix},
  T_{2} \begin{pmatrix} b\\ \bar{e}\end{pmatrix}}\leq R_{2}.
\end{equation*}
We now integrate with respect to $b$ and $e$.
To compute the integral
\begin{align*}
V(a,f,c)&=\int\limits_{\scprod{\begin{pmatrix} b\\ \bar{e}\end{pmatrix},
  T_{2} \begin{pmatrix} b\\ \bar{e}\end{pmatrix}}\leq R_{2}} 
V(a,f,b,c,e)
 \dint (b,e)\\
&=\int\limits_{\btop{\scprod{\begin{pmatrix} b\\ \bar{e}\end{pmatrix},
  T_{2} \begin{pmatrix} b\\ \bar{e}\end{pmatrix}}\leq R_{2}}{f_{2}a_{2}\geq c^{2}}} 
\frac{2^{6}\pi^{2}}{a_{2}^{3}}
  \gz{a_{2}f_{2}-\abs{c}^{2}}^{2}\det A_{3}
 \dint (b,e)
\end{align*}
  we substitute $\det A_{3}=R_{2}-
  \scprod{\begin{pmatrix} b\\ \bar{e}\end{pmatrix},
  T_{2} \begin{pmatrix} b\\ \bar{e}\end{pmatrix}}$
  and by Lemma \ref{le:integrateonellipsoid} we have
\begin{align*}
V(a,f,c)&=\frac{2^{6}\pi^{2}}{a_{2}^{3}}\gz{a_{2}f_{2}-\abs{c}^{2}}^{2}
\hskip-10pt\int\limits_{
 \scprod{\begin{pmatrix} b\\ \bar{e}\end{pmatrix},
  T_{2} \begin{pmatrix} b\\ \bar{e}\end{pmatrix}}\leq R_{2}} 
\hskip-10pt
 \gz{R_{2}-\scprod{\begin{pmatrix} b\\ \bar{e}\end{pmatrix},
   T_{2} \begin{pmatrix} b\\ \bar{e}\end{pmatrix}}}\dint (b,e)\\
&=\frac{2^{6}\pi^{2}}{a_{2}^{3}}\gz{a_{2}f_{2}-\abs{c}^{2}}^{2}\times
 \frac{\pi^{2}R_{2}^{3}}{6\det T_{2}},
\end{align*}
  where $\det T_{2}=\det A_{2}$, therefore
\begin{equation*}
V(a,f,c)=
\left\{\begin{array}{cll}
\di \frac{2^{5}\pi^{4}}{3}\gz{a_{2}f_{2}-\abs{c}^{2}}^{2}\times\gz{\det A_{2}}^{2}
  &\mbox{if} & f_{2}a_{2}\geq \abs{c}^{2},\\
\di 0  &\mbox{if} & f_{2}a_{2}<\abs{c}^{2}.
\end{array}
\right.
\end{equation*}

In the final step we assume that the parameters $a,f$ are given and 
  consider the requirement $\det A_{2}\geq 0$.
It means that $\abs{c}^{2}\leq af$, therefore if 
\begin{equation*}
\abs{c}^{2}\leq\min\kz{af,(1-a)(1-f)}
\end{equation*}
  then
\begin{equation*}
V(a,f,c)=\frac{2^{5}\pi^{4}}{3}\gz{a_{2}f_{2}-\abs{c}^{2}}^{2}\times\gz{af-\abs{c}^{2}}^{2}.
\end{equation*}
If $a+f\leq 1$, then $af\leq(1-a)(1-f)$.
In this case using polar coordinates for $c$ we have
\begin{align}
V(a,f)&=2\pi\int\limits_{0}^{\sqrt{af}}\frac{2^{5}\pi^{4}}{3}\gz{a_{2}f_{2}-r^{2}}^{2}
  \gz{af-r^{2}}^{2}\times r\dint r \nonumber \\
&=\frac{2^{4}\pi^5}{45}a_{1}^{3}f_{1}^{3}
  (a_{1}^{2}f_{1}^{2}-5a_{1}a_{2}f_{1}f_{2}+10a_{2}^2f_{2}^2).
\label{eq:Vaf1}
\end{align}
If $a+f\geq 1$, then $af\geq(1-a)(1-f)$.
In this case using polar coordinates for $c$ we have
\begin{align}
V(a,f)&=2\pi\int\limits_{0}^{\sqrt{(1-a)(1-f)}}\frac{2^{5}\pi^{4}}{3}\gz{a_{2}f_{2}-r^{2}}^{2}
  \gz{af-r^{2}}^{2}\times r\dint r \nonumber\\
&=\frac{2^{4}\pi^5}{45}a_{2}^{3}f_{2}^{3}
  (a_{2}^2f_{2}^2-5a_{1}a_{2}f_{1}f_{2}+10a_{1}^{2}f_{1}^{2}).
\label{eq:Vaf2}
\end{align}
Equations \eqref{eq:Vaf1} and \eqref{eq:Vaf2} give back Equation \eqref{eq:Vaf}.
The volume of the space of quantum channels is
\begin{equation}
V=\int_{0}^{1}\int_{0}^{1}V(a,f)\dint a\dint f=\frac{2\pi^{5}}{4725}.
\end{equation}
\end{proof}

\subsection{Unital qubit channels}

The following parametrization of $\Q^1\subset\R^{9}$ is considered
\begin{equation}
\label{eq:matQunital}
Q=\begin{pmatrix}
a_{1} & b & c & d\\
\bar{b} & a_{2} & e & -c\\
\bar{c} & \bar{e} & a_{2} & -b\\
\bar{d} & -\bar{c} & -\bar{b} &a_{1} \end{pmatrix},
\end{equation}
  where $a_{1}\in\sz{0,1}$, $a_{2}=1-a_{1}$ and $b,c,d,e\in\C$.
Let us define $a=a_{1}$.
The underlying classical channel corresponding to these parameter values are given by 
$Q^{1}_{\cl}=\begin{pmatrix}
        a & 1-a \\
        1-a & a
       \end{pmatrix}$.
Let us choose the unitary matrix
\begin{equation}
\label{eq:unitunital}
U=\begin{pmatrix}
0&0&1&0\\
1&0&0&0\\
0&1&0&0\\
0&0&0&1 \end{pmatrix}
\end{equation}
  and define the matrix $A$ as
\begin{equation}
\label{eq:matAunital}
A=U^{*}QU=\begin{pmatrix}
a_{2}&e&\bar{b}&-c\\
\bar{e}&a_{2}&\bar{c}&-b\\
b&c&a_{1}&d\\
-\bar{c}&-\bar{b}&\bar{d}&a_{1}\end{pmatrix}, 
\end{equation}
  which is positive definite if and only if $Q$ is positive definite.

\begin{theorem}
 The volume of the space $\Q^{1}$ with respect to the Lebesgue measure is
\begin{equation*}
V(\Q)=\frac{8\pi^{4}}{945}
\end{equation*}
  and the distribution of volume over classical channels can be written as
\begin{equation}
\label{eq:Va}
V(a)=\frac{2^{4}\pi^{4}}{3}a^{4}(1-a)^{4}.
\end{equation}
\end{theorem}
\begin{proof}
Since there is an unitary transformation \eqref{eq:unitunital} between 
  the set of matrices of the form of \eqref{eq:matAunital} and the quantum channels,
  their volumes are the same.
We compute the volume of the set of matrices given by parameterization \eqref{eq:matAunital}.

The volume element corresponding to the parametrization \eqref{eq:matQunital} is 
  $2^{7}\dint\la_{12}$.
The matrix $A$ in Equation \eqref{eq:matAunital} if positive definite if and only if 
  $\det (A_i) > 0$ for $i=1,2,3,4$. 

First we assume that the parameter $a$ and the submatrix $A_{3}$ are given and 
  consider the requirement $\det A_{4}\geq 0$.
Simple calculation shows that we have 
\begin{equation*}
\det A_{4}=\frac{\gz{\det A_{3}}^{2}}{\det A_{2}}-\abs{d'}^{2}\det A_{2},
\end{equation*}
  where 
\begin{equation*}
d'=d+\frac{2bc(1-a)-\bar{e}c^2-b^2e}{\det A_{2}}.
\end{equation*}
In this case the inequality $\det A_{4}\geq 0$ can be written in the form of
\begin{equation*}
\abs{d'}\leq\frac{\det A_{3}}{\det A_{2}}.
\end{equation*}
The transformation $d\mapsto d'$ is a shift, therefore it does not change the 
  volume element.
So we have
\begin{equation*}
V(a,b,c,e)=\int\limits_{\abs{d'}\leq\frac{\det A_{3}}{\det A_{2}}} 2^{7} 
  \dint (d')
 =2^{7}\pi\gz{\frac{\det A_{3}}{\det A_{2}}}^{2}.
\end{equation*}

In the next step we assume that the parameter $a$ and the submatrix $A_{2}$ are given
  and consider the requirement $\det A_{3}\geq 0$.
We have 
\begin{equation*}
\det A_{3}=R_{2}-
  \scprod{\begin{pmatrix} \bar{b}\\ \bar{c}\end{pmatrix},
 T_{2} \begin{pmatrix} \bar{b}\\ \bar{c}\end{pmatrix}},
\end{equation*}
where
\begin{equation*}
R_{2}=A_{33}\det A_{2}
\quad\mbox{and}\quad
T_{2}=\begin{pmatrix}  a_{2} & -e\\ -\bar{e} & a_{2} \end{pmatrix}.
\end{equation*}
The inequality $\det A_{3}\geq 0$ can be written in the form of
\begin{equation*}
\scprod{\begin{pmatrix} \bar{b}\\ \bar{c}\end{pmatrix},
  T_{2} \begin{pmatrix} \bar{b}\\ \bar{c}\end{pmatrix}}\leq R_{2}.
\end{equation*}
We now integrate with respect to $b$ and $c$.
To compute the integral
\begin{align*}
V(a,e)&=\int\limits_{\scprod{\begin{pmatrix} \bar{b}\\ \bar{c}\end{pmatrix},
  T_{2} \begin{pmatrix} \bar{b}\\ \bar{c}\end{pmatrix}}\leq R_{2}} 
V(a,b,c,e)
 \dint (b,c)\\
&=\int\limits_{\scprod{\begin{pmatrix} \bar{b}\\ \bar{c}\end{pmatrix},
  T_{2} \begin{pmatrix} \bar{b}\\ \bar{c}\end{pmatrix}}\leq R_{2}}
 2^{7}\pi\gz{\frac{\det A_{3}}{\det A_{2}}}^{2}
 \dint (b,c)
\end{align*}
we substitute $\det A_{3}=R_{2}-
  \scprod{\begin{pmatrix} \bar{b}\\ \bar{c}\end{pmatrix},
 T_{2} \begin{pmatrix} \bar{b}\\ \bar{c}\end{pmatrix}}$ 
  and by Lemma \ref{le:integrateonellipsoid} we have
\begin{align*}
V(a,e)&=\frac{2^{7}\pi}{\gz{\det A_{2}}^{2}}
\int\limits_{
 \scprod{\begin{pmatrix} \bar{b}\\ \bar{c}\end{pmatrix},
  T_{2} \begin{pmatrix} \bar{b}\\ \bar{c}\end{pmatrix}}\leq R_{2}} 
 \gz{R_{2}-\scprod{\begin{pmatrix} \bar{b}\\ \bar{c}\end{pmatrix},
   T_{2} \begin{pmatrix} \bar{b}\\ \bar{c}\end{pmatrix}}}^{2}\dint (b,c)\\
&=\frac{2^{5}\pi^{3}a^{4}}{3}\times\det A_{2}.
\end{align*}

Finally we assume that the parameter $a$ is given and 
  consider the requirement $\det A_{2}\geq 0$.
The condition $\det A_{2}\geq 0$ means that $\abs{e}\leq 1-a$, therefore  
  using polar coordinates for $e$ we have
\begin{equation*}
V(a)=2\pi\int_{0}^{1-a}
  \frac{2^{5}\pi^{3}a^{4}}{3}\times((1-a)^{2}-r^{2})\times r\dint r
  =\frac{2^{4}\pi^{4}}{3}a^{4}(1-a)^{4},
\end{equation*}
  which gives back Equation \eqref{eq:Va}.
The volume of the space of unital quantum channels is 
\begin{equation*}
V=\int_{0}^{1}V(a)\dint a=\frac{8\pi^{4}}{945}.
\end{equation*}
\end{proof}

One might think about the generalization of the presented results,
  although in a more general setting several complications occur.
For example, in the case of unital qubit channels one should integrate
  on the Birkhoff polytope, which would cause difficulties since even
  the volume of the polytope is still unknown \cite{IgorPak}.

\section{State transformations under random channels}

In this point, we study how qubits transform under uniformly distributed 
  random quantum channels with respect to the Lebesgue measure.
For simplification in this Section uniformly means that uniformly 
  with respect to the Lebesgue measure.

For further calculations, we need the Pauli basis representation of qubit channels.
The Pauli matrices are the following.
\begin{equation*}
\sigma_{1}=\begin{pmatrix} 0&1\\ 1& 0\end{pmatrix} \qquad
\sigma_{2}=\begin{pmatrix} 0&\ci \\ -\ci& 0\end{pmatrix} \qquad
\sigma_{3}=\begin{pmatrix} 1&0\\ 0&-1\end{pmatrix}
\end{equation*}
We use the Stokes representation of qubits which gives a bijective correspondence
  between qubits and the unit ball in $\R^{3}$ via the map
\begin{equation*}
\kz{x\in\R^{3}\vert\ \norm{x}_{2}\leq 1}\to \M{2}\qquad 
  x\mapsto \frac{1}{2}\gz{I+x\cdot\sigma},
\end{equation*}
  where $\di x\cdot\sigma=\sum_{j=1}^3 x_i \sigma_i$ and
  $\di I=\begin{pmatrix} 1&0\\ 0&1\end{pmatrix}$.
The vector $x$, which describes the state called Bloch vector and the unit ball in
  this setting is called Bloch sphere.

Any trace-preserving linear map $Q:\C^{2\times 2}\to\C^{2\times 2}$ can be written
  in this basis as
\begin{equation*}
 Q\gz{\frac{1}{2}(I+x\cdot\sigma)} = \frac{1}{2}\gz{I+(v+Tx)\cdot{\sigma}},
\end{equation*}
  where $v\in\R^3$ and $T$ is a $3\times 3$ real matrix.
Necessary and sufficient condition for complete positivity of such maps 
  are presented in \cite{RuskaiSzarekWerner}.
If the Choi matrix of qubit channel is given by Equation \eqref{eq:matQ},
  then the Pauli basis representation has the following form.
\begin{equation}
\label{eq:matvT}
v=\begin{pmatrix} \Rea (b+g) \\ -\Ima (b+g) \\ a+f-1 \end{pmatrix}  \quad
T=\begin{pmatrix}
  \Rea(d+e) & \Ima (d+e) & \Rea (b-g)  \\
  \Ima(e-d) & \Rea(d-e)  & \Ima (g-b)  \\
   2\Rea (c) & 2\Ima (c)  & a-f
 \end{pmatrix}
\end{equation}

The next lemma expresses the simple fact that uniformly distributed qubit-qubit
  channels have no preferred direction according to the Stokes parameterization of
  the state space.

\begin{lemma}\label{lem:J}
An orthogonal orientation preserving transformation $O$ in $\R^3$ induces maps 
  $\alpha_O,\beta_O:\Q\to\Q$ via Stokes parametrization $\alpha_O (Q) = O\circ Q$ 
  and $\beta_O (Q) = Q\circ O$. 
The Jacobian of these transformations are $1$.
The Jacobian of the restricted transformations 
  $\alpha'_{O}=\alpha_{O}\bigm\vert_{\Q^{1}}$ and 
  $\beta'_{O}=\beta_{O}\bigm\vert_{\Q^{1}}$ are $1$.
\end{lemma}
\begin{proof}
We used a computer algebra program to verify this lemma.
We considered three different kind of rotations according to the plane of rotations
  ($xy$, $xz$ and $yz$ plane).
It is enough to prove that the Jacobian of the generated $\alpha,\beta$ 
  transformation is $1$, since every orthogonal orientation preserving transformation 
  can be written as a suitable product of these elementary rotations.
We present the calculations for $\beta_{0}$, where
\begin{equation*}
O=\begin{pmatrix} 
  1 & 0 & 0\\
  0& \cos\alpha & -\sin\alpha\\
  0&\sin\alpha & \cos\alpha 
\end{pmatrix} \quad \alpha\in\R. 
\end{equation*}
If we consider a quantum channel given by parameters as in Equation \eqref{eq:matQ},
  then the effect of $\beta_{O}$ can be computed 
\begin{align*}
\beta_{O}(a,f,b_{1},b_{2},c_{1},c_{2},&d_{1},d_{2},e_{1},e_{2},g_{1},g_{2})\\
 & =(a',f',b'_{1},b'_{2},c'_{1},c'_{2},d'_{1},d'_{2},e'_{1},e'_{2},g'_{1},g'_{2}),
\end{align*}
  where all the parameters are real numbers and subscript $1$ refers to the real part
  and $2$ to the imaginary part.
We list some of the new parameters.
\begin{align*}
a'&=\frac{a+f}{2}+\frac{(a-f)\cos\alpha}{2}+c_{2}\sin\alpha\\
f'&=\frac{a+f}{2}-\frac{(a-f)\cos\alpha}{2}-c_{2}\sin\alpha\\
b'_{1}&=\frac{b_{1}(1+\cos\alpha)+g_{1}(1-\cos\alpha)}{2}
  +\frac{(e_{2}+d_{2})\sin\alpha}{2}\\
b'_{2}&=\frac{b_{2}(1+\cos\alpha)+g_{2}(1-\cos\alpha)}{2}+
  +\frac{(e_{2}-d_{2})\sin\alpha}{2}\\
c'_{1}&=c_{1}\\
c'_{2}&=c_{2}\cos\alpha-\frac{(a-f)\sin\alpha}{2}
\end{align*}
Next we computed the $12\times 12$ coefficient matrix, which is the derivative of
  the function $\beta_{0}$, and the computed determinant of the coefficient matrix turned
  to be $1$.
The similar computation was done for the other rotations.
The Jacobian of transformations $\alpha_{0}$, $\alpha'_{O}$ and $\beta'_{O}$ 
  was checked similarly.
\end{proof}

The idea of calculations about state transformations under random quantum channel is 
  presented by the following simpler case.

\begin{theorem}
\label{th:uniformrandomtomostmixed}
Applying uniformly random channels to the most mixed state, 
  the radii distribution of the resulted quantum states is the following.
\begin{equation}
\label{eq:uniformrandomtomostmixed}
\kappa(r)=40r^2(1-r)^6 (r^{3}+6r^{2}+12r+2) \quad r\in\sz{0,1}
\end{equation}
\end{theorem}
\begin{proof}
Applying a quantum channel of the form of given by Equation \eqref{eq:matQ} 
  to the most mixed state gives $z$ component $z'=a+f-1$.
If $z\geq 0$ then take the (not normalized) distribution from Equation \eqref{eq:Vaf}
\begin{equation*}
\tilde{V}(a,f)=(1-a)^3(1-f)^3\gz{(1-a)^{2}(1-f)^{2}-5a(1-a)f(1-f)+10a^{2}f^{2}}.
\end{equation*}
The density function of the $z$ component comes from the integral
\begin{equation*}
\eta(z)\sim\int\limits_z^1 \tilde{V}(a,z+1-a)\dint a.
\end{equation*}
The $z<0$ case can be handled in a similar way.
After normalization we have the following formula for the density function.
\begin{equation}
\label{eq:zdist}
\eta(z)=\frac{20}{11}\gz{z^4+7\abs{z}^3+17z^2+7\abs{z}+1}(1-\abs{z})^7 \quad z\in\sz{-1,1}
\end{equation}	

The distribution of quantum channels is invariant for orthogonal transformations 
  (Lemma \ref{lem:J}, the Jacobian of $\alpha_O$ is $1$). 
This means, that for every orthogonal basis the distribution of the $z$ component 
  of the image of the maximally mixed state is given by Equation \eqref{eq:zdist}.
Using Lemma \ref{lem:rv} we have
\begin{equation*}
\kappa(r)=-2r\eta'(r),
\end{equation*}
  which gives the desired formula for $\kappa$ immediately.
\end{proof}

It is worth to note that contrary to the classical case in quantum setting
  the entropy of the most mixed state will decrease after a random quantum channel
  is applied, since the Bloch radius of the resulted quantum state is
  $\di\frac{50}{143}$ in average.

Now we study the effect of unital uniformly distributed quantum channels.

\begin{theorem}
Assume that uniformly distributed unital quantum channel is applied to a given state with 
  Bloch radius $r_{0}$.
The radii distribution of the resulted quantum states is the following.
\begin{equation}
\label{eq:kappai}
\kappa_{1}(r,r_{0})=\frac{315}{16}\times\frac{r^{2}(r_{0}^{2}-r^{2})^{3}}{r_{0}^{9}}
  \chi_{\sz{0,r_{0}}}(r)
  \quad r\in\sz{0,1}.
\end{equation}
\end{theorem}
\begin{proof}
Since the distribution of unital uniform quantum channels is invariant for orthogonal 
  transformations (Lemma \ref{lem:J}, the Jacobian of $\beta'_{O}$ is $1$), we can assume
  that the initial state was given by the vector $(0,0,r_{0})$
  ($r_{0}\in\left]0,1\right]$).
Applying a unital quantum channel of the form of \eqref{eq:matQunital} to the
  initial state, we get a state with $z$ component $z'=r_{0}(2a-1)$.
The density function of the parameter $a$ of uniformly distributed unital quantum
  channels is a normalized form of \eqref{eq:Va}
\begin{equation*}
\Tilde{V}(a)=630a^{4}(1-a)^{4}\quad a\in\sz{0,1}.
\end{equation*}
If $z\in\sz{-1,1}$ arbitrary, then
\begin{equation*}
P(z'<z)=P\gz{a<\frac{z+r_{0}}{2r_{0}}}
=\left\{\begin{array}{cll}
\di 0 & \mbox{if} & z\leq -r_{0},\\
\di \int\limits_{0}^{(z+r_{0})/(2r_{0})} \Tilde{V}(a)\dint a& \mbox{if} 
  & -r_{0}<z< r_{0},\\
\di 1 & \mbox{if} & z\geq r_{0}.\\
\end{array}
\right.
\end{equation*}
We have for the density function of the $z$ component 
\begin{equation*}
f_{r_{0}}(z)=\frac{\dint P(z'<z)}{\dint z}
  =\frac{315}{256}\times\frac{(r_{0}^{2}-z^{2})^{4}}{r_{0}^{9}}\chi_{\sz{-r_{0},r_{0}}}(z),
\end{equation*}
  where $\chi$ denotes the characteristic function.
If the distribution of the $z$ component is known then by Lemma \ref{lem:rv}
  we can compute the radial distribution which gives us Equation \eqref{eq:kappai}.
\end{proof}

\begin{figure}[ht]
\centering{
  \includegraphics[width = 0.6\textwidth]{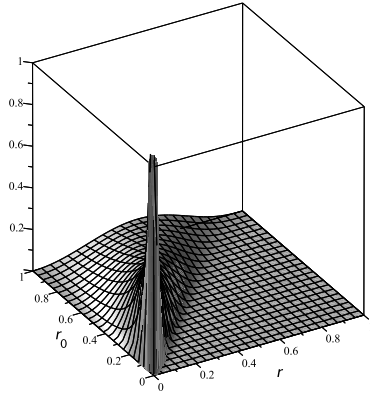}
}
\caption{
The function $\kappa_{1}(r,r_{0})$.
That is the radii distribution ($r$) of the resulted quantum states if
  uniformly distributed unital quantum channels were applied to a given state with 
  Bloch radius $r_{0}$.} 
\label{fig:kappa1}
\end{figure}

The transition probability between different Bloch radii under uniformly distributed
  unital quantum channels $\kappa_{1}(r,r_{0})$ is shown in Figure \ref{fig:kappa1}.
As it is expected, a unital quantum channel decreases the initial Bloch radius $r_{0}$,
  the new Bloch radius is $\di\frac{63}{128}r_{0}$ in average.
Since $\di r'\sim\frac{r}{2}$, repeated application of uniformly distributed unital 
  quantum channels maps every initial state to the most mixed state and
  the convergence is exponential.

\begin{theorem}
Assume that uniformly distributed quantum channel is applied to a given state with 
  Bloch radius $r_{0}$.
The radii distribution of the resulted quantum states is the following.
\begin{equation}
\label{eq:kappa}
\kappa(r,r_{0})=\left\{\begin{array}{l}
\mbox{If $0<r\leq r_{0}$:}\\
\di \frac{40r^{2}}{r_{0}(1+r_{0})^{6}}
   (21r^{4}-6r^{2}r_{0}^{2}-36r^{2}r_{0}+r_{0}^{4}+6r_{0}^{3}+12r_{0}^2+2r_{0}),\\[1em]
\mbox{if $r_{0}<r\leq 1$:}\\
\di \frac{40r(r-1)^{6}}{(1-r_{0}^{2})^{6}}
   (21r_{0}^4-6r^{2}r_{0}^{2}-36rr_{0}^2+r^4+6r^3+12r^2+2r).
\end{array}\right.
\end{equation}
\end{theorem}
\begin{proof}
Since the distribution of unital quantum channels is invariant for orthogonal 
  transformations (Lemma \ref{lem:J}, the Jacobian of $\beta'_{O}$ is $1$) we can assume
  that the initial state was given by the vector $(0,0,r_{0})$
  ($r_{0}\in\left]0,1\right]$).
Applying a quantum channel of the form of \eqref{eq:matQ} to the
  initial state, we get a state with $z$ component $z'=a+f-1+r_{0}(a-f)$.
The density function of parameters $a,f$ of uniformly distributed quantum
  channels is a normalized form of \eqref{eq:Vaf}
\begin{equation*}
\tilde{V}(a,f)=
\left\{ \begin{array}{cll}
\di V_{u}(a,f)  & \mbox{if} & 1\leq a_{1}+f_{1}, \\
\di V_{l}(a,f)  & \mbox{if} & 1> a_{1}+f_{1},
\end{array}
\right.
\end{equation*}
  where
\begin{align*}
V_{u}(a,f)&=840(1-a)^{3}(1-f)^{3}
  ((1-a)^{2}(1-f)^{2}-5a(1-a)f(1-f)+10a^{2}f^{2})\\
V_{l}(a,f)&=840a^{3}f^{3}
  (a^2f^2-5a(1-a)f(1-f)+10(1-a)^{2}(1-f)^{2}).
\end{align*}

First, we compute the probability $P(z'<\xi)$, where $\xi\in\sz{-1,1}$ is an
  arbitrary parameter.
To determine the probability $P(z'<\xi)$ the solution of the inequality
\begin{equation*}
z'=a+f-1+r_{0}(a-f)<\xi
\end{equation*}
  is needed for every parameter $r_{0}\in\left]0,1\right]$ and $\xi\in\sz{-1,1}$, taking
  into account the constraints $0\leq a,f\leq 1$.

To simplify this computation we define temporarily
\begin{align*}
&a_{1}=\frac{1+\xi}{1+r_{0}}, \quad
a_{2}=\frac{\xi+r_{0}}{1+r_{0}}, \quad
f_{1}=\frac{1+\xi}{1-r_{0}}, \quad
f_{2}=\frac{\xi-r_{0}}{1-r_{0}}\\
&q=\frac{1+r_{0}}{1-r_{0}}, \quad\mbox{and}\quad
x_{0}=\frac{\xi+r_{0}}{2r_{0}}.
\end{align*}

In the $\xi<-r_{0}$ case to compute the probability $P(z'<\xi)$ we have to integrate
  the density function $\tilde{V}(a,f)$ over the marked area shown 
  in Figure \ref{fig:Ineqaulityi}.
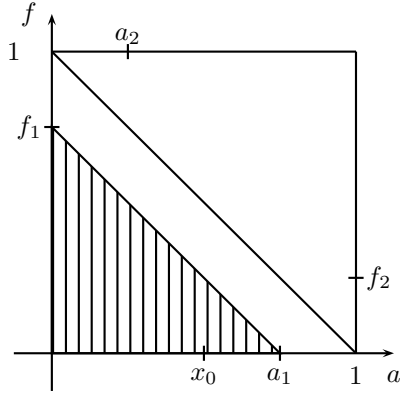
\begin{figure}[ht]
\centering{
\psset{xunit=1cm,yunit=1cm}
\begin{pspicture}(6,6)
\psline{->}(0.5,1)(5.5,1) % x axes
\rput(5.5,0.7){$a$}       % x axes label
\psline{->}(1,0.5)(1,5.5) % y axes
\rput(0.7,5.5){$f$}       % y axes label
\psline{-}(5,1)(5,5)      % unit square
\psline{-}(1,5)(5,5)      %
\rput(5,0.7){$1$}         %
\rput(0.5,5){$1$}         %
\psline{-}(1,5)(5,1)
\psline{-}(4,0.9)(4,1.1) \rput(4,0.7){$a_{1}$}
\psline{-}(3,0.9)(3,1.1) \rput(3,0.7){$x_{0}$}
\psline{-}(2,4.9)(2,5.1) \rput(2,5.2){$a_{2}$}
\psline{-}(0.9,4)(1.1,4) \rput(0.7,4){$f_{1}$}
\psline{-}(4.9,2)(5.1,2) \rput(5.3,2){$f_{2}$}
\pspolygon[fillstyle=vlines,hatchangle=0](1,1)(4,1)(1,4) 
\end{pspicture}
}
\caption{Solution of the inequality $z'<\xi$ in the $\xi<-r_{0}$ case.} 
\label{fig:Ineqaulityi}
\end{figure}

\begin{align*}
&P(z'<\xi)=\int_{0}^{a_{1}}\int\limits_{0}^{f_{1}-aq} V_{l}(a,f)\dint f\dint a=\\
&\frac{(10\xi^4-88\xi^2r_{0}^2+495r_{0}^4-80\xi^3+704\xi r_{0}^2+228\xi^2-198r_{0}^2-144\xi+33)(1+\xi)^8}
  {66(1-r_{0}^{2})^6}
\end{align*}

In the $-r_{0}\leq\xi\leq r_{0}$ case to compute the probability $P(z'<\xi)$ 
  we have to integrate the density function $\tilde{V}(a,f)$ over the four marked areas
  shown in Figure \ref{fig:Ineqaulityii}.
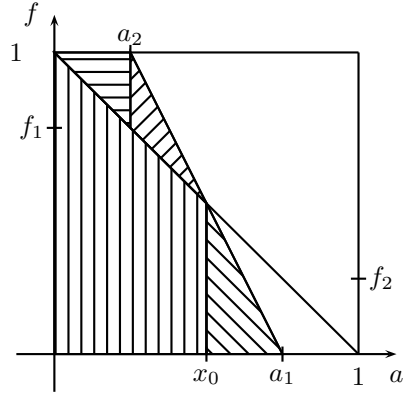
\begin{figure}[ht]
\centering{
\psset{xunit=1cm,yunit=1cm}
\begin{pspicture}(6,6)
\psline{->}(0.5,1)(5.5,1) % x axes
\rput(5.5,0.7){$a$}       % x axes label
\psline{->}(1,0.5)(1,5.5) % y axes
\rput(0.7,5.5){$f$}       % y axes label
\psline{-}(5,1)(5,5)      % unit square
\psline{-}(1,5)(5,5)      %
\rput(5,0.7){$1$}         %
\rput(0.5,5){$1$}         %
\psline{-}(1,5)(5,1)
\psline{-}(4,0.9)(4,1.1) \rput(4,0.7){$a_{1}$}
\psline{-}(3,0.9)(3,1.1) \rput(3,0.7){$x_{0}$}
\psline{-}(2,4.9)(2,5.1) \rput(2,5.2){$a_{2}$}
\psline{-}(0.9,4)(1.1,4) \rput(0.7,4){$f_{1}$}
\psline{-}(4.9,2)(5.1,2) \rput(5.3,2){$f_{2}$}
\psline{-}(2,5)(4,1)
\psline{-}(3,3)(3,1)
\pspolygon[fillstyle=vlines,hatchangle=0](1,1)(3,1)(3,3)(1,5) 
\pspolygon[fillstyle=vlines,hatchangle=45](3,1)(4,1)(3,3) 
\pspolygon[fillstyle=vlines,hatchangle=90](1,5)(2,5)(2,4)
\pspolygon[fillstyle=vlines,hatchangle=135](2,5)(2,4)(3,3)
\end{pspicture}
}
\caption{Solution of the inequality $z'<\xi$ in the $-r_{0}\leq\xi\leq r_{0}$ case.} 
\label{fig:Ineqaulityii}
\end{figure}
That is
\begin{align*}
P(z'<\xi)=&\int\limits_{0}^{x_{0}}\int\limits_{0}^{1-a}V_{l}(a,f) \dint f \dint a
 +\int\limits_{x_{0}}^{a_{1}}\int\limits_{0}^{f_{1}-aq} V_{l}(a,f) \dint f \dint a\\
 &+\int\limits_{0}^{a_{2}}\int\limits_{1-a}^{1} V_{u}(a,f)\dint f\dint a
 +\int\limits_{a_{2}}^{x_{0}}\int\limits_{1-a}^{f_{1}-aq} V_{u}(a,f\dint f\dint a
\end{align*}
which gives us
\begin{align*}
P&(z'<\xi)=\frac{-1}{66r_{0}(1+r_{0})^{6}} 
(660\xi^7-396\xi^5r_{0}^2+220\xi^3r_{0}^4-100\xi r_{0}^6-33r_{0}^7\\
&-2376\xi^5r_{0}+1320\xi^3r_{0}^3-600\xi r_{0}^5-198r_{0}^6+2640\xi^3r_{0}^2-1440\xi r_{0}^4-495r_{0}^5\\
&+440\xi^3r_{0}-1640\xi r_{0}^3-660r_{0}^4-720\xi r_{0}^2-495r_{0}^3-120\xi r_{0}-198r_{0}^2-33r_{0}).
\end{align*}

Finally in the $\xi>r_{0}$ case to compute the probability $P(z'<\xi)$ we have to integrate
  the density function $\tilde{V}(a,f)$ over the marked area shown in 
  Figure \ref{fig:Ineqaulityiii}.
\begin{figure}[ht]
\centering{
\psset{xunit=1cm,yunit=1cm}
\begin{pspicture}(6,6)
\psline{->}(0.5,1)(5.5,1) % x axes
\rput(5.5,0.7){$a$}       % x axes label
\psline{->}(1,0.5)(1,5.5) % y axes
\rput(0.7,5.5){$f$}       % y axes label
\psline{-}(5,1)(5,5)      % unit square
\psline{-}(1,5)(5,5)      %
\rput(5,0.7){$1$}         %
\rput(0.5,5){$1$}         %
\psline{-}(4,0.9)(4,1.1) \rput(4,0.7){$a_{1}$}
\psline{-}(3,0.9)(3,1.1) \rput(3,0.7){$x_{0}$}
\psline{-}(2,4.9)(2,5.1) \rput(2,5.2){$a_{2}$}
\psline{-}(0.9,4)(1.1,4) \rput(0.7,4){$f_{1}$}
\psline{-}(4.9,2)(5.1,2) \rput(5.3,2){$f_{2}$}
\pspolygon[fillstyle=vlines,hatchangle=0](1,1)(5,1)(5,2)(2,5)(1,5)
\end{pspicture}
}
\caption{Solution of the inequality $z'<\xi$ in the $\xi>r_{0}$ case.} 
\label{fig:Ineqaulityiii}
\end{figure}
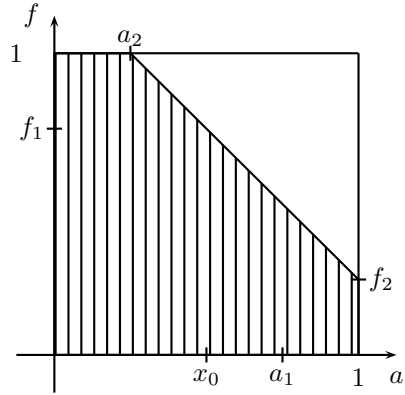

\begin{align*}
&P(z'<\xi)=1-\int\limits_{a_{2}}^{1}\int\limits_{f_{1}-aq}^{1} V_{u}(a,f)
  \dint f\dint a=1-\\
&\frac{(10\xi^4-88\xi^2r_{0}^2+495r_{0}^4+80\xi^3-704\xi r_{0}^2+228\xi^2-198r_{0}^2+144\xi+33)(1-\xi)^8}
  {66(1-r_{0}^{2})^6}
\end{align*}

Now we can compute the density function of the $z$ component as
\begin{equation*}
f_{z}(\xi)=\frac{\dint P(z'<\xi)}{\dint \xi}.
\end{equation*}
Since the density function is even ($f_{z}(\xi)=f_{z}(-\xi)$), we consider only the
  $\xi\geq 0$ case.
If $\xi>r_{0}$, we have
\begin{equation*}
f_{z}(\xi)=\frac{20(1-\xi)^7}{33(1-r_{0})^{6}}
  (3\xi^4-22\xi^2r_{0}^2+99r_{0}^4+21\xi^3-154\xi r_{0}^2+51\xi ^2-22r_{0}^2+21\xi+3)
\end{equation*} 
and if $0\leq x\leq r_{0}$, then
\begin{align*}
f_{z}(\xi)=&\frac{-10}{33r_{0}(1+r_{0})^{6}}
  (231\xi^6-99\xi^4r_{0}^2+33\xi^2r_{0}^4-5r_{0}^6-594\xi^4r_{0}+198\xi^2r_{0}^3\\
&-30r_{0}^5+396\xi^2r_{0}^2-72r_{0}^4+66\xi^2r_{0}-82r_{0}^3-36r_{0}^2-6r_{0}).
\end{align*}
Now we have the distribution of the $z$ component and by Lemma \ref{lem:rv}
  we can get the radial distribution $\kappa$ \eqref{eq:kappa}.
\end{proof}

The transition probability between different Bloch radii under uniformly distributed
  channel $\kappa(r,r_{0})$ is shown in Figure \ref{fig:kappav}.
\begin{figure}[ht]
\centering{
  \includegraphics[width = 0.6\textwidth]{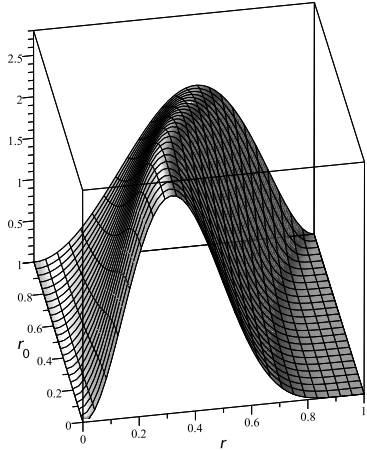}
}
\caption{
The function $\kappa(r,r_{0})$.
That is the radii distribution ($r$) of the resulted quantum states if
  uniformly distributed quantum channels were applied to a given state with 
  Bloch radius $r_{0}$.} 
\label{fig:kappav}
\end{figure}

Note that the function $\kappa(r,0)$ gives back the formula 
  \eqref{eq:uniformrandomtomostmixed} in Theorem \ref{th:uniformrandomtomostmixed}.
In Figure \ref{fig:average} the average Bloch radius is shown after uniformly 
  distributed random quantum channel applied to a state with Bloch radius $r_{0}$.
From this figure it is clear that if the initial Bloch radius is small then 
  a quantum channel likely increases the Bloch radius and if $r_{0}$ is big then
  decreases.
Repeated application of such kind of random quantum channels will send initial
  states to the Bloch radius $r\approx 0.388$.

\begin{figure}[ht]
\centering{
  \includegraphics[width = 0.4\textwidth]{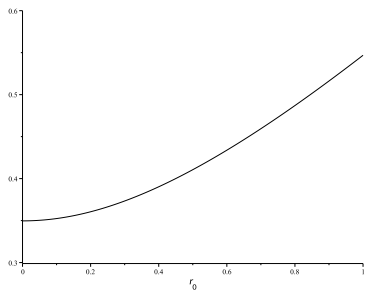}
}
\caption{The average Bloch radius after uniformly distributed random quantum channel
applied to a state with Bloch radius $r_{0}$.} 
\label{fig:average}
\end{figure}

\section{Concluding remarks}

In this work we considered the Choi representation of quantum channels and the
  Lebesgue measure on matrix elements.
We computed the volume of quantum channels and studied the effect of uniformly 
  randomly distributed (with respect the Lebesgue measure) general and unital
  qbit-qbit quantum channels using the Choi's representation.
It was shown that the chosen measure on the space of qbit-qbit channels is
  unitary invariant with respect to the initial and final qbit spaces separately.
We presented the Bloch radii distributions of states after a uniformly random 
  general or unital quantum channel was applied to a given state.
This gives opportunity to study the distribution of different information 
  theoretic quantities (for example different channel capacities, entropy gain,
  entropy of channels etc.) and the effect of repeated applications of uniformly
  random channels.

%\bibliography{ref} 

\begin{thebibliography}{10}

\bibitem{AndaiVol}
A.~Andai.
\newblock Volume of the quantum mechanical state space.
\newblock {\em Journal of Physiscs A: Mathematical and Theoretical},
  39:13641--13657, 2006.

\bibitem{BoudaVarga}
J.~Bouda, M.~Koniorczyk, and A.~Varga.
\newblock Random unitary qubit channels: entropy relations, private quantum
  channels and non-malleability.
\newblock {\em The European Physical Journal D}, 53(3):365--372, 2009.

\bibitem{KZycz}
W.~Bruzda, V.~Cappellini, H.-J. Sommers, and K.~Zyczkowski.
\newblock Random quantum operations.
\newblock {\em Physics Letters A}, 373(3):320--324, 2009.

\bibitem{Choi}
M.-D. Choi.
\newblock Completely positive linear maps on complex matrices.
\newblock {\em Linear Algebra and Appl.}, 10:285--290, 1975.

\bibitem{HarrowA}
A.~Harrow, P.~Hayden, and D.~Leung.
\newblock Superdense coding of quantum states.
\newblock {\em Phys. Rev. Lett.}, 92(18), 2004.

\bibitem{Jamiolkowski}
A.~Jamio\l~kowski.
\newblock Linear transformations which preserve trace and positive
  semidefiniteness of operators.
\newblock {\em Rep. Mathematical Phys.}, 3(4):275--278, 1972.

\bibitem{NielsenChuang}
M.~A. Neilsen and I.~L. Chuang.
\newblock {\em Quantum Computation and Quantum Information}.
\newblock Cambridge University Press, Cambridge, 2000.

\bibitem{Omkar}
S.~Omkar, R.~Srikanth, and Subhashish Banerjee.
\newblock Dissipative and non-dissipative single-qubit channels: dynamics and
  geometry.
\newblock {\em Quantum Information Processing}, 12(12):3725--3744, 2013.

\bibitem{IgorPak}
Igor Pak.
\newblock Four questions on birkhoff polytope.
\newblock {\em Annals of Combinatorics}, 4(1):83--90, 2000.

\bibitem{AronPasieka}
A.~Pasieka, D.~W. Kribs, R.~Laflamme, and R.~Pereira.
\newblock On the geometric interpretation of single qubit quantum operations on
  the bloch sphere.
\newblock {\em Acta Appl. Math.}, 108(697), 2009.

\bibitem{PetzQinf}
D.~Petz.
\newblock {\em Quantum Information Theory and Quantum Statistics}.
\newblock Springer, Berlin-Heidelberg, 2008.

\bibitem{RuskaiSzarekWerner}
M.~B. Ruskai, S.~Szarek, and E.~Werner.
\newblock An analysis of completely positive trace- preserving maps on
  $\mathcal{M}_2$.
\newblock {\em Linear Algebra Appl.}, 347:159--187, 2002.

\end{thebibliography}
%\bibliographystyle{plain}

% Copy from chvol.bbl

\def\polhk#1{\setbox0=\hbox{#1}{\ooalign{\hidewidth
  \lower1.5ex\hbox{`}\hidewidth\crcr\unhbox0}}}

\end{document}